\documentclass{sig-alternate-2013}

\usepackage{graphicx}
\usepackage{subfigure}
\usepackage{epstopdf}
\usepackage{amsmath}
\usepackage{footnote}
\usepackage{color}

\usepackage{times}
\newcommand{\subparagraph}{}

\usepackage[small,compact]{titlesec}
\usepackage[labelfont=bf,font=bf]{caption}

\newcommand{\squishlist}{
 \begin{list}{$\bullet$}
  { \setlength{\itemsep}{0pt}
     \setlength{\parsep}{3pt}
     \setlength{\topsep}{3pt}
     \setlength{\partopsep}{0pt}
     \setlength{\leftmargin}{1.5em}
     \setlength{\labelwidth}{1em}
     \setlength{\labelsep}{0.5em} } }
\newcommand{\squishend}{
  \end{list}  }

\makesavenoteenv{tabular}
\makesavenoteenv{table}

\newcommand{\specialcell}[2][c]{%
  \begin{tabular}[#1]{@{}c@{}}#2\end{tabular}}
\graphicspath{{images/}}

\def\test{\textrm{test}}
\def\train{\textrm{train}}

\def\Vtest{V_\textrm{test}}
\def\Vtrain{V_\textrm{train}}

\def\ego{\mathbf{Ego}}

\newcommand{\Exp}[1]{
E\left[{#1}\right]
}
\newtheorem{lemma}{Lemma}
\newtheorem{theorem}{Theorem}

\permission{\copyright 2017 International World Wide Web Conference Committee \\ (IW3C2), published under Creative Commons CC BY 4.0 License.}
\conferenceinfo{WWW'17 Companion,}{April 3--7, 2017, Perth, Australia.}
\copyrightetc{ACM \the\acmcopyr}
\crdata{978-1-4503-4914-7/17/04. \\
http://dx.doi.org/10.1145/3041021.3055139 \\
\includegraphics{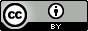}
 }

\begin{document}

\numberofauthors{4}
\author{
\alignauthor
Rahmtin Rotabi\\
       \affaddr{Cornell University}\\
       \email{rahmtin@cs.cornell.edu}
\alignauthor
Krishna Kamath\\
       \affaddr{Twitter Inc.}\\
       \email{kkamath@twitter.com}
\alignauthor
Jon Kleinberg\\
       \affaddr{Cornell University}\\
       \email{kleinber@cs.cornell.edu}
\and  
\alignauthor
Aneesh Sharma\\
       \affaddr{Twitter Inc.}\\
       \email{aneesh@twitter.com}
}

\title{Detecting Strong Ties Using Network Motifs}

\maketitle

\begin{abstract}
     Detecting strong ties among users in social and information networks
  is a fundamental operation that can improve performance on a
  multitude of personalization and ranking tasks. There are a variety
  of ways a tie can be deemed ``strong'', and in this work we use a
  data-driven (or supervised) approach by assuming that we are
  provided a sample set of edges labeled as strong ties in the
  network. Such labeled edges are often readily obtained from the
  social network as users often participate in multiple overlapping
  networks via features such as following and messaging. These
  networks may vary greatly in size, density and the information they
  carry --- for instance, a heavily-used dense network (such as the
  network of followers) commonly overlaps with a secondary sparser
  network composed of strong ties (such as a network of email or phone
  contacts). This setting leads to a natural strong tie detection
  task: given a small set of labeled strong tie edges, how well can
  one detect unlabeled strong ties in the remainder of the network?

  This task becomes particularly daunting for the Twitter network due
  to scant availability of pairwise relationship attribute data, and
  sparsity of strong tie networks such as phone contacts. Given these
  challenges, a natural approach is to instead use structural network
  features for the task, produced by {\em combining} the strong and
  ``weak'' edges. In this work, we demonstrate via experiments on
  Twitter data that using only such structural network features is
  sufficient for detecting strong ties with high precision. These
  structural network features are obtained from the presence and
  frequency of small network motifs on combined strong and weak
  ties. We observe that using motifs larger than triads alleviate
  sparsity problems that arise for smaller motifs, both due to
  increased combinatorial possibilities as well as benefiting strongly
  from searching beyond the ego network. Empirically, we observe that
  not all motifs are equally useful, and need to be carefully
  constructed from the combined edges in order to be effective for
  strong tie detection. Finally, we reinforce our experimental
  findings with providing theoretical justification that suggests why
  incorporating these larger sized motifs as features could lead to
  increased performance in planted graph models.
\end{abstract}

%


\section{Introduction}
Many large social media platforms are constructed so that
users participate in multiple networks simultaneously.
On Twitter, for example, a user can establish links to other users
by following them, or sending them a direct message, or by including them in
a phone book or e-mail address book.
Each of these modes of interaction defines a different network on the
set of users --- while these networks are clearly related, they
represent different types of connections; you might easily follow
someone whom you have no expectation of ever messaging or contacting
by phone.  Moreover, the networks can differ greatly in their
usage and sparsity; Twitter users will generally follow multiple other users,
but many may have never populated their phone books.

In these social media contexts, we thus encounter a wide range of
cases in which a network $G_L$ with a large number of links coexists
with an overlapping, but much sparser network $G_S$. As in the case of
the (dense) follower graph and (sparser) phone book graph on Twitter,
the sparsity of $G_S$ often comes about for two reasons.  First, in
contrast to the follower graph, the phone book is not the main feature
of the site.  Second, even for users who have created links in the
sparser graphs $G_S$ (such as the phone-book graph), these links tend
to correspond to the user's {\em strong ties}, and so are less
numerous.  Despite their sparsity, networks like the 
phone-book graph contain information that is immensely valuable, in part
because of their focus on strong ties; they can be used for improving
relevance and personalization in user/content recommendations, 
creating more personalized notifications, and other
applications. Many of these tasks can benefit from estimates of 
edges likely to belong to $G_S$, even if they haven't been
explicitly reported by users.

Thus, estimating which edges belong to these types of sparse graphs $G_S$ 
can be viewed as a 
type of strong-tie detection problem \cite{gilbert_2012}.
But as we discuss next, our setting adds new aspects to the problem;
existing techniques for strong-tie detection produce weak performance
in our context, and in contrast we develop a set of new methods that
yield significantly more powerful results.

\paragraph{\bf Data-driven strong tie detection}
Our particular setting of the Twitter network adds some new dimensions
to strong tie detection that we believe haven't been studied before
in concert:
(i) a data-driven formulation of the problem;
(ii) extreme sparsity of demographic features; and
(iii) possible ineffectiveness of interaction features.
The first point here is that while the sparse graphs $G_S$ that
we consider have significant overlap with the strong-tie structure,
they are not precisely the set of strong ties, and so we need to
take a data-driven approach: rather than starting from sociological
principles about strong-tie structure, we use the 
existing structure of
edges around nodes that participate in $G_S$ to learn the features
that are most predictive,
The second and third points, about the limitations of demographic
and interaction features for our problem, make clear why new
techniques are needed.  Recall that these kinds of features
(reciprocity of interaction, engagement volume, and related measures) 
have been proven
to be the most effective at strong tie detection in prior
work~\cite{gilbert_2012}. But in our setting, the demographic features
are quite sparse as users do not need to report demographic
information to use the website. Furthermore, due to the
data-driven approach, interaction features may not be most indicative
of the particular strong tie label we're trying to predict.

As a concrete way to see the challenge, we began by running standard
strong-tie detection algorithms using as many features
from existing methods~\cite{gilbert_2012} 
as we could (all the high-weight features were
available), trying to predict whether an edge in the Twitter follower graph
was also in the phone-book graph.  We found that on a
fully balanced dataset these features only provide a prediction
accuracy of $56\%$ which barely beats the random baseline by $6\%$.
As we will show in this work, we are able to do much better with the
techniques we develop here; our eventual performance will be $87\%$
accuracy for the same task. Given the poor performance of prior
baselines, it becomes important to understand the challenges in this
particular version of the problem more clearly.

In the motivating applications for the problem, the graph $G_S$
providing the strong tie labels is sparse because many (or most) users
have not yet started using the feature defining $G_S$, leading to a
graph with a large fraction of disconnected or isolated nodes. And
even the nodes that are not isolated have very low degrees.  As a
result, standard link-prediction methods cannot be applied on $G_S$
alone to ``bootstrap'' the internal structure of $G_S$ by itself.

Our setting is closer to the ``cold-start'' problem in recommendation
systems
\cite{bohb-cold-start,spup-cold-start,ssbxc-cold-start},
where the challenge is to make recommendations to users for
whom the system has no history at all, and the 
general idea is to use some source of side information to
provide recommendations to such users. This is where the weak ties
provided by the denser network $G_L$ come into play; even if a
user $u$ has no edges in $G_S$, the weak ties that $u$ has in $G_L$
provide information about the presence of $u$'s edges in $G_S$. We
note that in contrast to prior work on the cold-start problem, our
side information in this case is a network in its own right. In this
respect, our strong tie detection problem is a question of information
transfer. In contrast to existing work on information transfer,
however, we are considering settings where we may well not have any a
priori principles informing the typical patterns of links in our
domain. Thus, we aim to induce patterns purely from training data,
adapting to different settings independently of whether any particular
hypothesized structural patterns turn out to be the most effective.

The core of our problem, then, is to combine information from the weak
ties in $G_L$ with the strong ties from the subset of nodes that
participate in $G_S$, to predict strong ties
for nodes that currently do not participate in $G_S$.
We now discuss some of our
techniques and results for this problem.

\paragraph*{\bf Filling in a sparsely-populated network}
Addressing this question requires that we use information latent in
both $G_S$ and $G_L$.  Intuitively, from the portions of $G_S$ that we
are able to observe, we try to infer the typical patterns formed by
the edges in $G_S$ and their interplay with edges in $G_L$.  We then
go to nodes where the edges of $G_S$ are absent, and we look for
evidence of similar patterns among the edges of $G_L$.  This provides
evidence for where the hidden edges of $G_S$ might be.  For this
strategy to work, it is important that $G_L$ provide us with
sufficient information about $G_S$; a key aspect of the approach is
based on transferring information between the two graphs $G_L$ and
$G_S$.  In particular, for the ``patterns'' formed among the edges of
$G_S$, we use the presence of small subgraphs or {\em motifs} that are
highly represented in the observable portion of $G_S$.  Given a user
$a$ who is incident to no edges of $G_S$ --- a user with no strong
ties --- we try to infer which of $a$'s incident edges (or weak ties)
in $G_L$ are most likely to belong to $G_S$ as well.  To do this, we
assign each of $a$'s incident edges in $G_L$ multiple scores based on
their participation in certain small subgraphs.  Note that the edges
of these subgraphs may have mixed membership in $G_L$ and $G_S$; in
particular, if $(a,b)$ is an edge of $G_L$, then $b$ may have incident
edges in $G_S$ even though $a$ does not.

Using the subgraph-induced scores for an edge $(a,b)$ as a vector of
features, we can then learn a classifier for labeling edges of $G_L$
as strong ties based on the nodes in the graph that have incident
edges from $G_S$.  In this way, we are not presupposing which
particular motifs are indicative of membership in $G_S$, but instead
identifying the most important motifs from the data.

In evaluations on the Twitter user network, we obtain strong
performance in inferring the presence of unobserved labels in several
of the site's main underlying networks, using {\em only} network
features. In particular, adopting an evaluation framework in which we
hide the presence of a subset of the direct messages, phone book, and
email address book edges --- thus providing ground truth for our
evaluation --- we find that our approach using the mutual following
graph as $G_L$ yields high accuracy in detecting strong ties.

Moreover, among the graph motifs that are most important for the task
of classifying edges of $G_S$, we find that motifs on more than three
nodes play a crucial role.  This forms an intriguing contrast to the
work on strong-tie detection within networks of weak ties --- one of
the most widely-studied cases of information transfer between
different types of networks --- since in that context the key issue
has traditionally been the participation of an edge $(a,b)$ in
triangles, which correspond in our framework to features comprised
entirely of a set of three-node subgraphs.  For the networks we infer
here, the participation of edges $(a,b)$ in larger motifs turns
out to be vital as well.

\paragraph*{\bf The role of larger graph motifs}
A possible reason for the role of larger graph motifs in our task is
that these larger structures provide a partial antidote to the problem
of sparsity --- due to the number of combinatorial possibilities, they
have the potential to be more abundant than triangles in the training
data we have on $G_S$, particularly when the observed part of $G_S$ is
highly sparse.  To understand the trade-off between sparsity and the
use of larger structures as features, we propose and analyze a set of
generative network models where this effect appears clearly, and with
provable guarantees.  Specifically, we consider a mathematical model
in which nodes belong to planted communities, and edges of $G_S$ lie
within these communities.  As in prior work on random graphs with
planted community structure, we cannot directly observe community
membership (a proxy for strong ties), but the structure of the graph
conveys latent information about it.  Our question, however, is
related to but different from the standard problem of inferring
community membership; rather, we want to classify edges by their
membership in $G_S$, where the random generation of $G_S$ is based on
the community structure.

We discover that the qualitative findings from our evaluation on the
Twitter dataset hold for this generative model as well, and with
provable guarantees arising from analysis of the model; motifs larger
than triangles provably help in identifying edges of $G_S$, and the
gain from these subgraphs increases as the observable portion of $G_S$
becomes sparser.  Moreover, our approach based on the frequency of
motifs is robust enough that it even yields guarantees when nodes
belong to multiple overlapping communities with independent
membership.

Overall, the interaction of the computational evaluation on Twitter
and the analysis of the generative models suggests that our concrete
task, transferring information from one graph to discover strong ties
around isolated nodes in another, is a useful and general problem that
provides insight into the role of local network motifs.

In the remainder of the paper, we first present an experimental
evaluation of our methods on the Twitter dataset, followed
by modeling and simulations that seek to provide a
theoretical basis for understanding our experimental observations.

\section{Related Work}
Our approach is related to several lines of research concerned with
network structure, particularly in the domain of social and
information networks. More specifically, our work can be seen as
positioned at the intersection of three topics in network analysis:
(a) strong tie detection (albeit under extreme sparsity of features),
(b) link prediction in a case where many nodes are isolated (because
they do not yet participate in $G_S$),
and (c) graph information transfer in a purely data-driven manner,
i.e. without pre-supposing any sociological basis.

As we discussed in the introduction, previous work on strong tie
detection~\cite{gilbert_2009,gilbert_2012,marsden-tie-strength} yields
only limited effectiveness in
our setting both due to feature sparsity (even of small
network motifs such as triads) and due to our data-driven labeling goal;
we seek to transfer information from a denser graph on the
same set of nodes. Existing methods for strong-tie detection have used
structural information based on triangles, following ideas from
sociology
\cite{granovetter-weak-ties,marsden-tie-strength,eagle-inferring-friendship,jones-tie-strength,sintos-tsaparas-strong-ties},
whereas we make use of motifs larger than triangles.
Other studies that
have looked at information transfer include analysis of
advisor-advisee relationships \cite{wang-advisor-advisee}, 
romantic relationships
\cite{backstrom-cscw14}, and other types of relationship transfer
\cite{liu-heterogeneous,sun-heterogeneous,tang-heterogeneous}.
Our goal, however, is to transfer information without a priori sociological
principles to guide the process.


A number of fields have developed frameworks for analyzing small
subgraphs that occur frequently in larger networks; such formalisms
have been termed {\em network motifs} \cite{milo-motifs}, the {\em
  triad census} \cite{faust-triad-2010}, and {\em frequent subgraph
  mining} \cite{kuramochi-freq-subgraph,yan-han-freq-subgraph}, and
have been used in extremal graph theory \cite{borgs-graph-hom}
and social media analysis \cite{ugander-www13}.
Our work takes a different perspective as its starting point; rather
than using the frequency of a given subgraph as the key criterion, we
are proceeding in a more supervised fashion, using data to infer which
subgraphs are most informative for our task.

Another relevant line of work is the network completion
problem \cite{kim2011modeling, kim2011network}. Results in this area
propose generative models recreating an entire graph to infer 
missing parts of a given partial graph.  
These methods have also been
developed primarily without the goal of 
handling a large fraction of isolated nodes; and
they only scale
up to hundreds of thousand of nodes, which is much smaller than our
graphs of interest.

The transfer of graph information bears a
distant relationship to transductive learning
\cite{joachims-transductive} and label propagation
\cite{zhu-label-prop}, although these methods
typically propagate node
attributes from labeled to unlabeled nodes, while we use
the network structure to infer information about edges.
The structure of
our generative models draws motivation
from {\em stochastic block models}
\cite{condon2001algorithms,mcsherry-planted,krzakala2013spectral,mossel2014belief}; for us, such models provide
a setting in which phenomena
we observe in our computational evaluation appear
with provable guarantees, providing a certain level of qualitative
insight into how they operate.

%

\section{Methods and Experiments}

We now describe our methods for predicting strong ties
by combining information from two different graphs $G_L$
and $G_S$, and we discuss the
results of experiments on the Twitter graph.

\subsection{Prediction Task}

We begin with the formal set-up for our prediction task.  Recall that
we are given two graphs: a denser graph $G_L=(V,E_L)$ that contains
all the edges available to us for analysis;
and an overlapping, sparsely populated graph
$G_S=(V,E_S)$, which contains the reported strong ties.
We will think of the edges in $E_S$ as being {\em labeled}
with their strong-tie status.
Recall that one of our primary motivations is to detect strong-tie
labels for users who have not reported any strong ties. However, to
set up a formal prediction task, we need a ground-truth labeled set
that we can compare our predictions to. Thus, we propose the following
evaluation framework to evaluate methods for predicting strong ties.
We pick a small set of test nodes $\Vtest \subset V$ and remove all of
their edges in $E_S$ while retaining their edges in $E_L$. Doing this
process will simulate the setting where we do not have any strong ties
reported around the test nodes.  The prediction task for a supervised
machine learning algorithm is to build a model on the remaining graph
that is able to predict for each $v\in \Vtest$, which of $v$'s edges
in $E_L$ was a removed edge from $E_S$.

Let $d_L(v)$ denote the degree of node $v$ in $G_L$, and let
$d_S(v)$ denote the degree of $v$ in $G_S$.
Our test set of nodes $\Vtest$ is a random subset of the nodes in $V$
subject to the condition that $d_S(v) > 0$ for all $v \in \Vtest$. We define
$\Vtrain = V - \Vtest$. We call the induced graph of $G_i$
(again, $i\in \{L,S\}$) on $\Vtrain$ and $\Vtest$ to be $G_{i,\train}$
and $G_{i,\test}$, respectively.

To reflect a scenario in which nodes in $\Vtest$ are users who have
not adopted the feature that grants us $E_s$, we allow the algorithm
to only utilize $G_{L,\train}$ and $G_{S,\train}$ for training a
model. Once the algorithm trains its model, it is given access to
$G_{L,\test}$ and the algorithm's prediction task is to output exactly
one edge for each node $v \in \Vtest$, namely the edge that is the
algorithm's best guess for belonging in $G_{S,\test}$. The precision
of the algorithm for a given set of test nodes is the fraction of
correct guesses; in traditional information retrieval parlance it is
known as precision in the first position, or $p@1$. Thus, we are
considering a setting in which a user joins the platform and links to
a few people in the heavily-used network $G_L$, but does not adopt the
sparsely-used network $G_S$; based on this, we wish to report a link that is
likely to be a strong tie for this user.

\subsection{Twitter Dataset}

We now describe the graphs that we analyze, all of which are produced
by interactions between users on Twitter.\footnote{All the
Twitter data has been analyzed in an anonymous, aggregated form to
preserve private information.}
We study four undirected graphs built from different
types of interactions between users on Twitter. The nodes in
these graphs are the Twitter users (note that a user may not correspond to
an individual, since other account types are possible on Twitter), 
and the interpretation of the edge depends on
the graph being considered. The graphs are as follows:
\squishlist
\item \textbf{Mutual Follow:} this is the graph of users who
  follow each other on Twitter. Since we focus on symmetric
  relationships on Twitter, this results in an undirected graph.
\item \textbf{Phone Book:} Twitter users often import their phone
  books for finding their friends on Twitter. That feature gives rise
  to a graph in which we have an edge between two users if and only if
  both users have imported their phone book and have each other's
  phone number.
\item \textbf{Email Address Book:} users can also import their email
  address book on Twitter, and analogous to the phone book, we can
  construct an email graph where an edge exists if both users have
  each other in their email address book.
\item \textbf{Direct Message:} Users on Twitter can also send private
  messages to each other, and these are called {\em direct
    messages}. Thus, we can consider an undirected graph in which
  there is an edge between each pair of users who
  have sent at least one direct message to each other.
\squishend

For our computational experiments, 
we collected a single complete snapshot of all
these graphs on July 30, 2015. The mutual follow graph contains 
hundreds of millions of users, and tens of billions of edges. The portion
of other graphs that we consider are approximately a tenth the size of
the mutual follow graph.\footnote{We are unable to share the exact
  size of these graphs but we note the relative sparsity of the strong
  tie graphs.} As the Mutual Follow graph is much more densely
populated than the others, we create three instances of our prediction
task: in each, the Mutual Follow graph forms $G_L$, and one of the
other graphs (Phone Book, Email Address Book, or Direct Message) forms
$G_S$.  We only preserve edges in $G_S$ that also belong to $G_L$;
this removes less than $0.5\%$ of edges, so it is a simplification
with negligible impact for our purposes.

We construct $\Vtest$ by randomly choosing $5\%$ of the
nodes that satisfy the following pair of conditions:
(i) they have an edge in $G_S$, and (ii) their degree in $G_L$ is between
$10$ and $75$. This degree restriction on $G_L$ reflects the motivating
scenario in which the users in $\Vtest$ are relatively lighter users (or new users)
of the platform, for whom we are trying to infer edges in networks
they have not yet populated (while still having enough edges to predict). The smallest $\Vtest$ over the three different graphs we look at 
includes roughly $29000$ nodes.

\subsection{Prediction Algorithms}

Our algorithms for predicting strong ties will operate as follows.
We start with a node $a \in \Vtest$ that has no incident edges in
$G_S$, and we would like to identify an edge in $G_L$, incident to $a$,
that in fact	 belongs to $G_S$.
Let $B$ be the set of neighbors of $a$ in $G_L$.
Because of the simplifying assumption, motivated above, that
$E_S \subseteq E_L$, the candidate edges incident to $a$ that we are
choosing among all consist of edges from $a$ to a node in $B$.

We develop a number of {\em score} functions, each of which assigns a
number to every node in $B$.  We can use each individual score as a
predictor in itself, by sorting the nodes in $B$ according to the
score, selecting the highest-scoring $b \in B$, and declaring $(a,b)$
to be the predicted edge in $E_S$.\footnote{We break ties based on
  degree (smaller is prioritized) and the time the user joined the
  platform (earlier is prioritized).}
 A few of these scoring functions
are used in the literatures on strong tie detection
and on link prediction, and we use
them as baselines. Recall from the introduction that interaction
features under-perform, so these are in fact stronger baselines for
the prediction task.  We also use the scores as features and combine
them using a machine-learning algorithm trained on the nodes in
$\Vtrain$ for which we have ground-truth edges belonging to $E_S$.  By
comparing the performance of different scores, as well as combinations
of them, we can thus determine which structures are most effective at
transferring information from $G_L$ to the unpopulated parts of
$G_S$. The fact that our scores include standard baselines from link
prediction and strong-tie prediction enables us to quantify
the gains from scores corresponding to more complex structural
formulations.

In order to define our scoring functions, we introduce the following notation.
Given a graph $G$ and a node $v$ in the graph, we use $N(G, v)$ to denote the
neighbors of $v$ in $G$ and the operator $\ego(G, v)$
to refer to the {\em ego network} (the induced subgraph on
$N(G,v) \cup \{v\}$) of the node $v$ in $G$.
Recall that we use $a$ for the node in $\Vtest$ for whom we
are predicting an edge in $E_S$, and $B = N(G,a)$.
We also use $C$ to denote the set of nodes at distance exactly two
from $a$ in $G_L$. The scoring functions we use are given in Table~\ref{Algorithms}.
We divide the scoring functions into two groups, Group 1 and Group 2.
In both groups, we also include a composite score that uses linear
regression to combine all the scores (and their logarithms) in the
group to form a single machine-learning classifier.  Scoring functions
in {\em Group 1} are standard benchmarks from link prediction and
strong-tie prediction, and they provide a baseline for comparison.
Scoring functions in {\em Group 2} use larger graph-theoretic
structures; where the Group 1 scores are based primarily on triangles
and counts of mutual neighbors, the scores in Group 2 use counts of
4-node structures (squares) and 5-node structures (pentagons), with
the formal specifications provided in Table~\ref{Algorithms}. We also
counted network motifs such as the complete graph of size $4$ and $5$
but they were too sparse and did not help the prediction.

\addtocounter{footnote}{1}
\begin{table*}[!htbp]
\caption{The definition of all the scoring algorithms evaluated experimentally.}
\small
\begin{tabular}{|c |c |c |c|}
  \hline
  Name & Definition & Optimization  & Group \\ \hline

  Random neighbor & \specialcell{A simplistic baseline that picks a
                    node in $B$ uniformly at random}& N/A & 1\\ \hline

  Lowest degree neighbor\footnotemark & \specialcell{Picks the node $b\in B$ with the
                           minimum degree $d_L(b)$ }& minimization & 1	\\ \hline

  Embeddedness & \specialcell{Picks the node $b \in B$ that has
                 maximum number of mutual neighbors with $a$}&  maximization & 1 \\ \hline

  Adamic-Adar &\specialcell{The score assigned to $b \in B$ is $\sum_{v \in \{B \cap N(G_L,b)\}} \frac{1}{log(d_L(v))}$
  		which is a weighted version of \\ Embeddedness giving a greater  weight to mutual neighbors with a lower degree} & maximization & 1\\ \hline

  H1 & \specialcell{This heuristic function assigns $d_L(b)$ as the
       score for $b\in B$ if $d_S(b) > 0$ and 0 otherwise} & minimization &  1\\ \hline

  Triangle & \specialcell{This method finds $N(G_S,b) \cap B$ for every $b \in B$}& maximization & 1 \\ \hline
 Basic ML model & \specialcell{This method builds an LR model using features from rows above}& maximization & 1 \\ \hline

  Square inside & \specialcell{For every $b \in B$, this algorithm
                  reports the number of cycles  with 4 nodes (or {\em
                  squares}) \\containing the edge $(a,b)$ with the
  other two nodes in the cycle also being in $B$.\footnotemark} & maximization &  2\\ \hline
 \addtocounter{footnote}{-1}

  Square outside & \specialcell{ For every $b \in B$, this function
                   reports the number of cycles with 4 nodes (or {\em
  squares}) \\containing the edge $(a,b)$,  exactly
  one other node in $B$ and one node in $C$.\footnotemark}& maximization & 2\\ \hline
 \addtocounter{footnote}{-1}

  Pentagon inside & \specialcell{ This method counts the number of
                    cycles of length $5$ for  each  $b \in B$ containing edge \\$(a,b)$ where all 5 nodes are inside $\ego(G_L,a)$.\footnotemark} & maximization & 2 \\ \hline
 \addtocounter{footnote}{-1}

  Pentagon outside & \specialcell{This method counts the number of
                     cycles of length $5$ for each  $b \in B$ containing edge \\$(a,b)$ where 3 nodes are inside $\ego(G_L,a)$ and the other nodes are in $C$.\footnotemark}& maximization & 2  \\ \hline
\addtocounter{footnote}{-1}

 Enhanced ML model & \specialcell{This method builds an LR model using all
             scores listed above as features}& maximization & 2 \\ \hline
\end{tabular}
\label{Algorithms}
\end{table*}
\addtocounter{footnote}{-1}

Before presenting our results, we illustrate the scoring functions
with a toy example. In Figure~\ref{fig: ToyExample}, the dashed edges
represent weak ties and the solid edges represent strong ties. Note
that edges labeled as strong ties also count as weak ties.
\begin{figure}[htp]
\centering
\includegraphics[width=0.9\columnwidth]{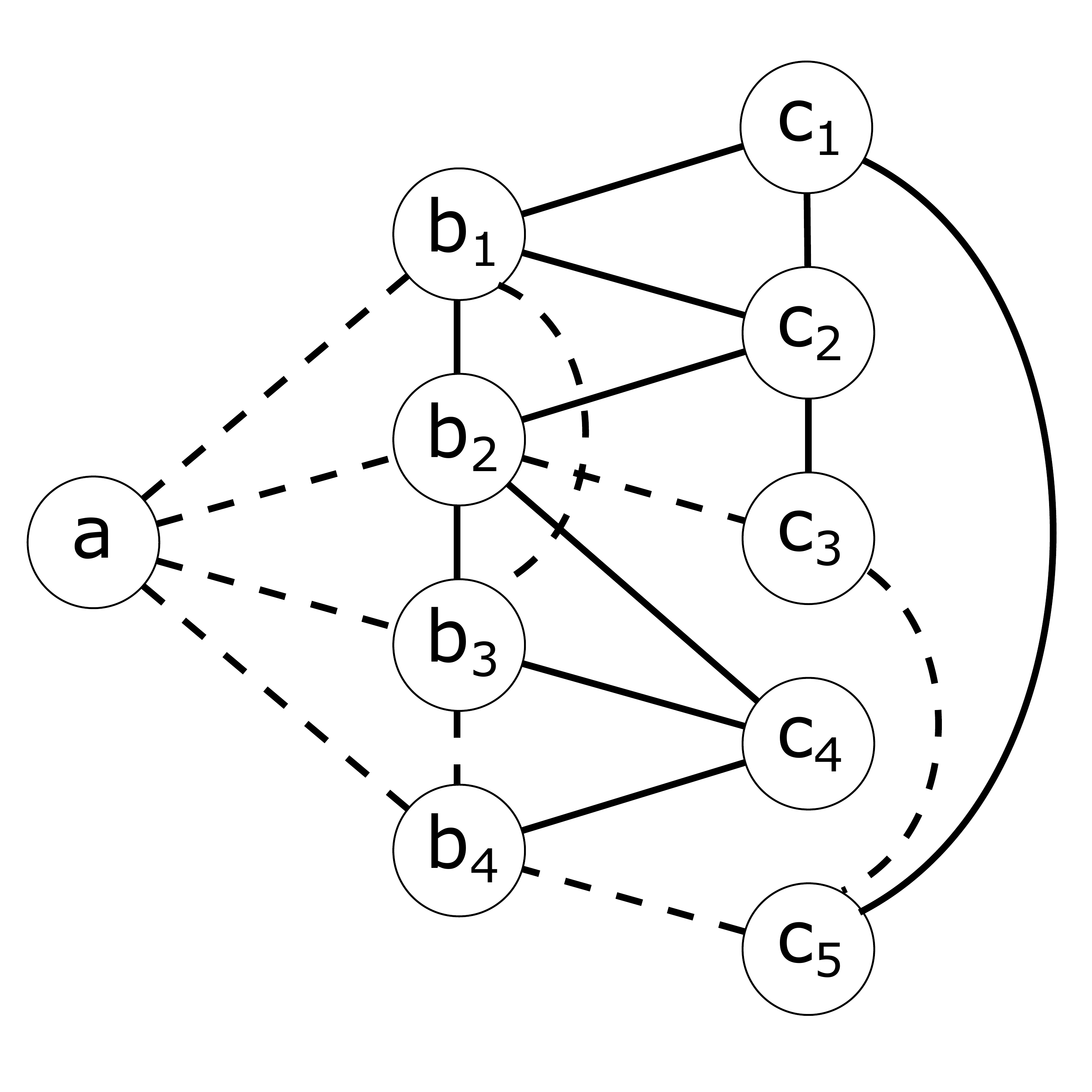}
\caption{A simple graph centered around node $a$.}
\label{fig: ToyExample}
\end{figure}
Consistent with the notation above we are looking at node $a$ and all the strong ties of $a$ are hidden.
The scores for node $b_1$ on this small graph are:
Degree=$5$, Embeddedness=$2$, Adamic-Adar= $\frac{1}{\log(5)} + \frac{1}{\log{(6)}}$,
H1=$5$, Triangles=$1$, Square Inside= $1$ (Square through $b_2$ and $b_3$),
Square Outside=$1$ (Square through $c_2$ and $b_2$), Pentagon Inside=$0$
and finally Pentagon Outside=$1$ (Pentagon through $c_1$, $c_2$ and $b_2$).

The focus of this work is on finding useful scoring functions for
prediction, but we also want to briefly mention computational
considerations for these. Computing small sub-structures such as
triangles, squares and pentagons on large graphs is well-known to be
a challenging problem; the triangles case in particular has a large
literature and many provably efficient approximations
\cite{seshadhri-pinar-kolda-triadic}. We briefly note that for
practical purposes, heuristics based on neighborhood sketches such as
HyperLogLog \cite{flajolet_2008} can be used to efficiently compute
these features in a MapReduce~\cite{dean2008mapreduce} computation.
We provide a brief description of this heuristic in section
\ref{sec:hyperloglog}. We emphasize however that in our experiments,
we do not use any approximations since our training/test data size was
small enough to be tractable with exact computation.


\begin{figure*}[!ht]
        \begin{minipage}[l]{0.57\columnwidth}
            \centering
            \includegraphics[width=\linewidth]{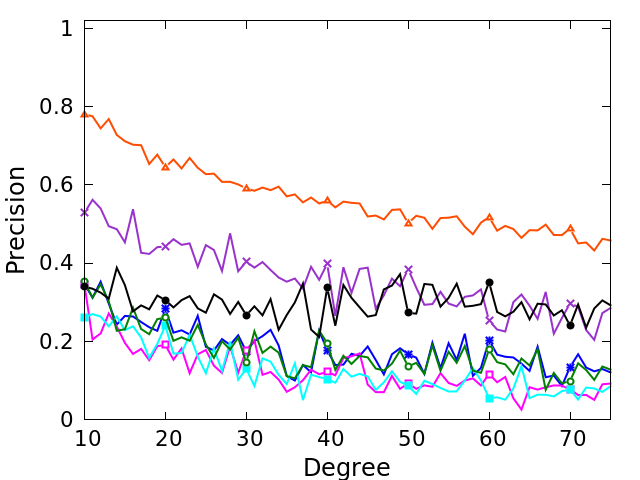}
           \captionsetup{labelformat=empty}
            \caption{(a) Direct Message}
        \end{minipage}
        \begin{minipage}[r]{0.57\columnwidth}
            \centering
            \includegraphics[width=\linewidth]{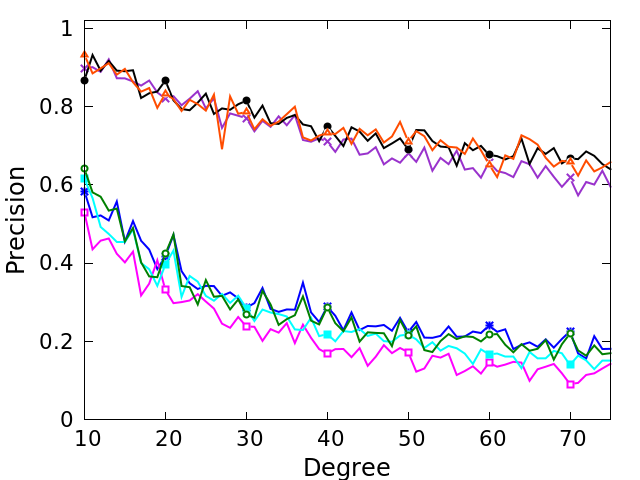}
            \captionsetup{labelformat=empty}
            \caption{(b) Phonebook}
        \end{minipage}
           \begin{minipage}[r]{0.82\columnwidth}
            \centering
            \includegraphics[width=\linewidth]{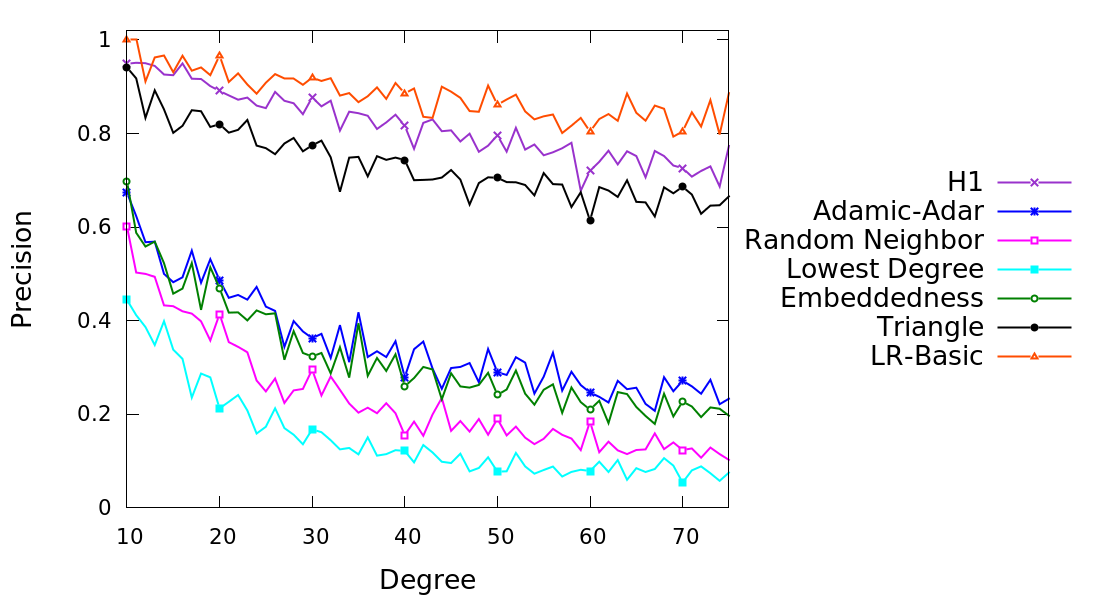}
            \captionsetup{labelformat=empty}
            \caption{\hspace{-1.1cm}(c) Email Address Book}
        \end{minipage}
        \addtocounter{figure}{-1}
        \addtocounter{figure}{-1}
        \addtocounter{figure}{-1}
  \caption{Precision at 1 for group 1 scoring functions}
 \label{Ego}
\end{figure*}

\begin{figure*}[!ht]
        \begin{minipage}[l]{0.57\columnwidth}
            \centering
            \includegraphics[width=\linewidth]{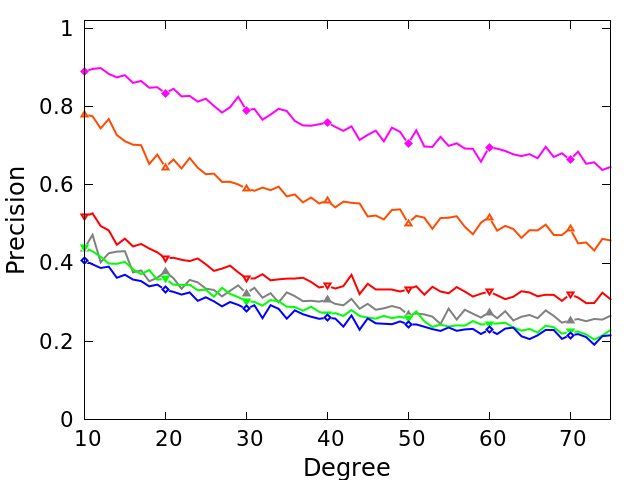}
           \captionsetup{labelformat=empty}
            \caption{(a) Direct Message}
        \end{minipage}
        \begin{minipage}[r]{0.57\columnwidth}
            \centering
            \includegraphics[width=\linewidth]{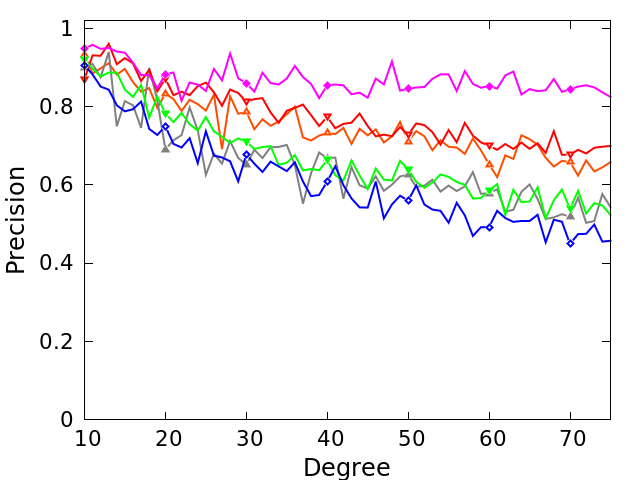}
            \captionsetup{labelformat=empty}
            \caption{(b) Phonebook}
        \end{minipage}
           \begin{minipage}[r]{0.82\columnwidth}
            \centering
            \includegraphics[width=\linewidth]{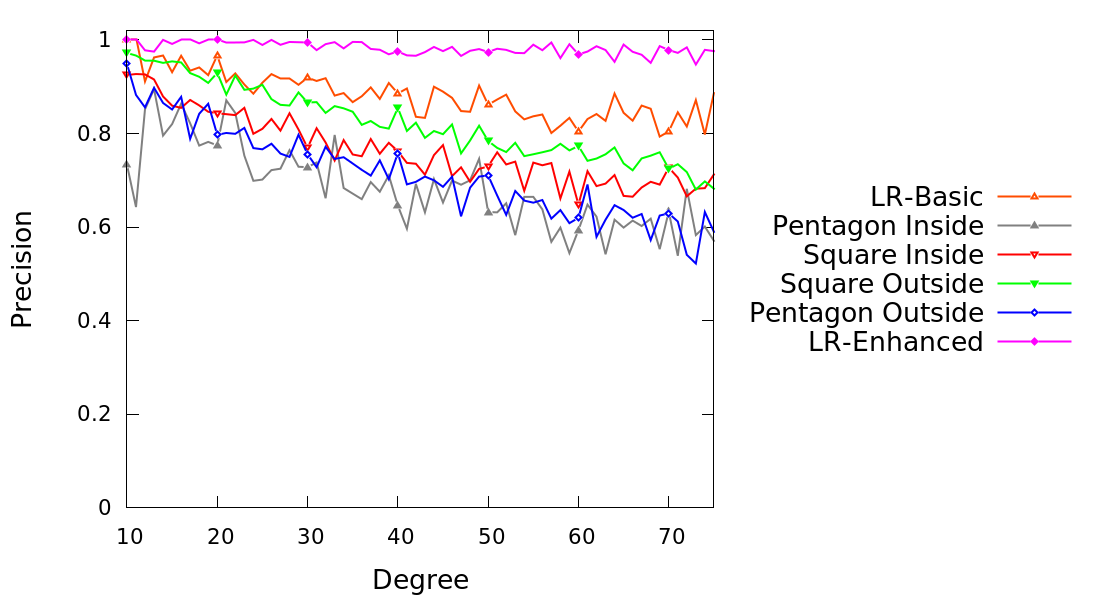}
            \captionsetup{labelformat=empty}
            \caption{\hspace{-1.1cm}(c) Email Address Book}
        \end{minipage}
        \addtocounter{figure}{-1}
        \addtocounter{figure}{-1}
        \addtocounter{figure}{-1}
  \caption{Precision at 1 for group 2 scoring functions and the linear model from group 1}
 \label{None-Ego}
\end{figure*}

\subsection{Performance on Prediction Task}
We evaluated each of the scoring algorithms presented in
Table~\ref{Algorithms} on the prediction task outlined above,
in which we seek to predict a single edge in $E_S$ incident to 
the given node $a$.
We present the results in terms of their {\em precision} --- the fraction
of instances on which the predicted edge indeed belongs to $E_S$.
Rather than providing just an overall precision number, we present our
results grouped according to the degree of the node $a$. This is to
reflect the fact that there is a gradation based on difficulty: nodes
of higher degree naturally form harder instances for some algorithms,
since there are more candidate edges to choose from (for instance, the
baseline of random guessing goes down correspondingly). On the other
hand, some algorithms are able to efficiently exploit more available
information from the larger neighborhood, and hence this presentation
highlights the dependence of our methods on degree.


The precision results are shown in Figures \ref{Ego} and
\ref{None-Ego}. Note that in these figures, we organize the algorithms
according to the Group classification in Table~\ref{Algorithms},
discussed above. As seen in figure \ref{None-Ego}, the Enhanced LR
model performs much better than the Simple LR model that combines all
the other baseline algorithms that were proposed previously in the
literature. We highlight that there is a large $17\%$, $12\%$ and
$10\%$ gain in precision in the direct message, phone book and Email
prediction task, respectively. 
And it is important to note that
we achieve these gains only by adding slightly larger network motifs
falling entirely inside or one step outside the ego-network. We emphasize the
raw precision value here as well --- a priori it would seem that a
precision of 75\% for phone book prediction using just the network
motifs would be too much to ask. The results speak for the power of
the motifs approach.  \footnotetext{This is similar to the IDF
  heuristic used in information retrieval.}
\addtocounter{footnote}{1} \footnotetext{The edges not incident to $a$
  in the cycles should be in $G_{S, train}$.}  

We now discuss further the 
performance of the algorithms in {\em Group 2}, depicted in
Figure~\ref{None-Ego}. These algorithms differ from the Group 1
benchmarks based on prior work in two aspects: (a) they look at richer
motifs in the ego network, such as squares and pentagons, and (b) they
also look beyond the ego-network for computing the scores. First, we
observe that looking at the richer motifs even in the ego-network is
fruitful as the {\em squares inside} and {\em pentagons inside}
algorithms outperform the triangle algorithm, albeit only
slightly. However, the results from looking beyond the ego-network, as
demonstrated by the {\em squares outside} and {\em pentagons outside}
algorithms, are mixed. Only in the email graph does the {\em squares
  outside} algorithm beat the others comprehensively. Intuitively,
looking beyond the ego-network results in a trade-off between sparsity
and signal strength: the signal strength of structures trails off with
increasing distance from the ego network, but the number of structures
found increases due to increasing combinatorial possibilities. From
the results, it seems that the square motifs potentially are a ``sweet
spot'' in this trade off, and we investigate this issue in the next
section in more depth. It is evident from the performance results
shown in Figure~\ref{None-Ego} that the LR models outperforms all the
other methods and the model using the Group 2 features outperforms the
model using Group 1 features from prior work. In examining the
enhanced LR model, we find that three significant features in the
model that we learned are \small $\textrm{log}(square\text{-}inside)$,
$\textrm{log}(square\text{-}outside)$ and $\textrm{log}(triangles)$.
\normalsize Since the LR model outperformed both these individual
features as well, this suggests that the model is combining
information from inside and outside the ego network. And this combined
algorithm achieves a much higher precision compared to individual
scores. Of course, the LR model has access to more features than any
individual score; but one might still not apriori expect the features
from inside and outside the ego network to complement each other
enough to achieve this level of increased precision.

We also note that results on the email graph are a bit different from
the other two. In particular, the lowest-degree neighbor algorithm is
worse than the random neighbor algorithm on the email graph, and the
square-outside algorithm performs better than the square-inside
algorithm (and all others for that matter). We hypothesize that this
might be due to the noisy nature of the email graph compared to direct
messages and phone book. Namely, presence in another user's email
address book is a much weaker signal than having a private
conversation or being in another user's phone book. Finally, we
observe that precision is often lower for high degree nodes, which is
perhaps surprising as high degree nodes have more connections/patterns
in their neighborhood that an algorithm could potentially
exploit. However, as noted at the outset,
there is also the added difficulty from having to choose from among a
larger set of neighbors, which is illustrated by the random
baseline. And the latter effect is clearly stronger, as the plots
indicate.

\subsection{Predicting Multiple Edges}

The prediction task formulated above focuses on producing a
single prediction. But one might suspect that if a scoring algorithm
performs well on this task, then it should also provide high precision
in the setting where multiple predictions are generated.  In this
section, we test this hypothesis experimentally. The task we set up
for the algorithm is as follows: we take a set of nodes from our test
set that have at least five edges in $G_S$ and given an algorithm's
scores, we report the five neighbors with the highest score to be its
prediction; thus measuring $p@5$.\footnote{Note that our ground
  truth data has only binary labels, and hence a ranking metric such
  as NDCG isn't applicable here.} The results of this experiment when
we set $G_S$ to be the phone book graph are shown in
Figure~\ref{PBPrecisonOn5}.


\begin{figure}[htp]
\centering
\includegraphics[width=\columnwidth]{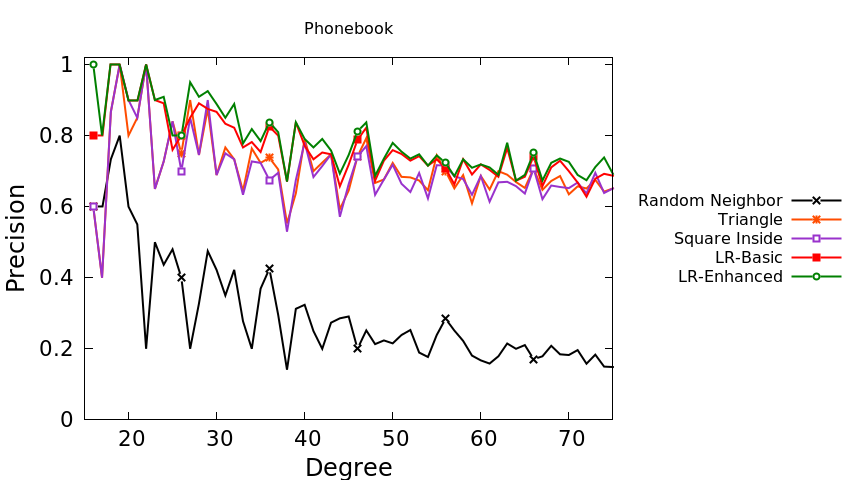}
\caption{Precision at 5 on the phone book graph for a selection of methods.
There is no node such that $d_L<16$ and $d_S\geq 10$, therefore
 the plot does not start from 10. The symbols (g1) and (g2) at the end of the legends show
 the group of the scoring function based on table \ref{Algorithms}.}
\label{PBPrecisonOn5}
\end{figure}

We briefly note that the results are quite consistent with the ones
for the prediction task: the Enhanced LR model outperforms the rest of
the algorithms and baselines. Hence, it seems reasonable to conclude
that relative performance on predicting even one edge is a good
indicator of performance on predicting multiple edges.

\section{Theoretical Modeling}
A key observation from our experimental results is that larger motifs
than triangles --- in particular, squares --- provide a substantial
performance boost on the strong tie prediction task. In this section,
we aim to provide theoretical insight into why this should be the
case. Specifically, in what settings can one expect richer motifs such
as squares to be more useful as features for strong tie prediction?

In order to study this question analytically, we define two simple
graph models that are motivated by stochastic block models. 
While the graphs produced by these models are much simpler than
what we encounter in practice, they serve as tractable structures
that capture some of the essential features of the networks 
in our applications.
Given a graph model, we can study
the usefulness of motifs by posing it as a feature sparsity
question: are squares a more discriminative feature than triangles
because of their higher prevalence? As we will see, this indeed turns out
to be true, and we will show via analytical results and simulations
that for certain parameter ranges in these graph models, squares are
more discriminative features than triangles.

To start with, we define a prediction task that reflects our
experimental evaluation (though it is not exactly equivalent). Given a
graph $G(V,E)$ with hidden labels $L = f(e) \rightarrow \{0,1\}$ for
all $e \in E$, an algorithm is required to predict the hidden labels
for edges. We will attempt to represent patterns for edges marked with
label $1$ via a graph generation model. The edges labeled as $1$ in
these models will act as strong ties and edges labeled as
$0$ will be the weak ties.

This task is quite similar to our empirical task, with a few differences.
First, recall that in the empirical task, 
we had to guess only one edge as being a strong tie. 
Here, however, we ask for an algorithm that labels all the
edges adjacent to the test node. Another difference between this model
and our prediction task is the fact that in the theoretical analysis
our algorithm does not have access to the labeled strong ties and will
treat all the ties as weak ties. We emphasize that the main purpose of
this section is to study the relative discriminative power of squares
and triangles features.



\subsection{Graph Models}
A key step in formulating a theoretical model is to have a graph model
that abstracts some of the important structural properties of 
friendship structures in real-world networks. Our approach is
to use ``planted community'' models so that there is a clear structure
that an algorithm can learn. Furthermore, in addition to the planted
structure, we also need to appropriately represent the sparsity and
noisiness of data represented in online ties of the real-world offline
networks. Hence, in all the models proposed here, the underlying
friendship structure in the graph is perturbed via sparsification and
by adding noisy edges to the graph.

\subsubsection{Single Planted Model}

We start with a basic model, which we term
the {\em single planted model}.  This is a stochastic block model graph
with $n$ nodes and $\lceil\frac{n}{c}\rceil$ communities each consisting of
$c$ nodes.  For parameters $p$, $q$, and $r$, the probability of an edge
between nodes in the same community is
$p\frac{q}{\sqrt{c}}$, and the probability of an edge between
nodes in different communities is $r$.

The graph generated by this model constitutes $G_L$, and the
subgraph consisting only of edges inside communities is $G_S$.
We can then imagine removing edges incident to a subset of
nodes $\Vtest$ in $G_S$, corresponding to the isolated nodes in our
prediction task.

We can think of the edge probabiltiy $p\frac{q}{\sqrt{c}}$ inside
communities arising as follows: we imagine $p$ as the probability
that each pair of nodes within a community knows each other, and
$\frac{q}{\sqrt{c}}$ as the probability that two such people in fact
form a link between each other on the platform.
This second filtering of the links via $\frac{q}{\sqrt{c}}$ makes
the task more realistic and more challenging.

\subsubsection{Double Planted Model}

{
This model is an extension of the single planted model. We create
two different random graphs $G_1$ and $G_2$ using the single planted
model. We call the subsets of the edges inside the community $G_{1_S}$
and $G_{2_S}$. The union of $G_{1_S}$ and $G_{2_S}$ will create our
graph $G_S$. Edges between two nodes that do not share a community are
added with probability $r$.
In this case each node is in exactly two communities.
The final graph created by this process will be $G_L$.
We call the communities in $G_L$ that
come from $G_{1_S}$ the Type-1 Group, and the communities that come
from $G_{2_S}$ the Type-2 Group.
}

\subsection{Theoretical Analysis}
Now we theoretically analyze the performance of squares and triangles
as predictors for finding edges in $G_S$ in the planted models proposed
above. Throughout this analysis, we will set the random noise edge
probability as $r=\frac{ \ln n }{n}$, and the group size to be
$c=\alpha_1 \ln n$. From this point on we define $\rho=pq$.
Due to similarity between the proofs for the Single and Double Planted Model we will
 state our claims  without proofs for the Single Planted Model, and
provide more details for the Double Planted Model.

We will show that in these models squares are on average a more discriminative
 feature than triangles. In particular, we will make this claim
by examining the gap between the expected value of the two features
for edges within a group versus edges between
groups. Given an edge $(x,y)$, denote the number of triangles and squares
 that the edge belongs to as $\Delta(x,y)$ and $\Box(x,y)$.

\subsubsection{Single Planted Model}

{
\begin{theorem}
For all edges $(x,y)$, we have $\Exp{\Delta(x,y)} <1$.
For edges that go between different communities we have \\ 
$\Exp{\Box(x,y)} <1$,
and for edges in the same community we have 
$\Exp{\Box(x,y)} > \left(1 -
  \frac{5}{n}\right)\sqrt{c}\rho^3  \simeq \sqrt{c}\rho^3.$
\end{theorem}
Thus, edges inside and between groups have a significant difference
in the expected number of squares they are involved in, roughly
equal to $\sqrt{c}\rho^3$. Hence, in the single planted model the squares
feature is expected to be more discriminative than triangles.

The functions $\Delta$ and $\Box$ have an integer range, and for these
features to be distinguishable they have to be non-zero, and the
probability of an edge outside a group being in a triangle or square
converges to zero.  So the power of the triangle/square feature is
equivalent to the probability of being in a triangle/square.  For an
an edge inside a group, the probability it is in a triangle is \small
$1-(1-(\frac{\rho}{\sqrt{c}})^2)^{c-2}$ \normalsize and the
probability it is in a square is \small
$1-(1-(\frac{\rho}{\sqrt{c}})^3)^{(c-2)(c-3)}$.  \normalsize The
latter is much larger for sparser (smaller $\rho$) graphs leading to
another indication that squares are a better feature than triangles.

\subsubsection{Double Planted Model}
{
\begin{theorem}
In the double planted model,
$\Exp{\Delta(x,y)} <1$ and $\Exp{\Box(x,y)} < 1$ for all edges
$(x,y)$ where $x$ and $y$ are not in the same community.
For edges $(x,y)$ where $x$ and $y$ are in the same community,
$\Exp{\Delta(x,y)} < 2$ and $\Exp{\Box(x,y)} > \sqrt{c}\rho^3$.
\end{theorem}

We start by stating two basic properties of the double planted model,
formalized in the following lemmas. We skip the proofs as they are
based upon a standard application of Chernoff bounds.

\begin{lemma}
  For a sufficiently large constant $\alpha_1$(\small$>\frac{4.1}{\delta^2}$\normalsize), with high
  probability(\small$>1-\frac{2}{cn}$\normalsize), each group's size is in \small $[(1 - \delta) c,(1 + \delta) c]$\normalsize.
\label{GroupSize}
\end{lemma}

\begin{lemma}
  For a sufficiently large constant $\alpha_1$(\small$>\frac{4.1}{\delta^2}$\normalsize), with high
  probability, no Type-1 group has an intersection of size more than
  $\alpha_2$(\small$\leq 3$\normalsize) with any Type-2 group.
\label{IntersectionSize}
\end{lemma}

The above two properties provide us with the tools to analyze the
expected values of the triangles and squares features.

\begin{lemma}
  For any nodes $x$ and $y$, we have $\Exp{\Delta(x,y)} \le 1$, except
  the case when $x$ and $y$ share the same group in both types, in which case
  $\Exp{\Delta(x,y)} \le 2$.  However, this latter event is rare.
\end{lemma}
\begin{proof}
These are the three cases we need to analyze:
\\
\textbf{x and y have two common groups:} The probability of
  $x$ and $y$ sharing both groups is $(\frac{c}{n})^2$, and
  conditioned on that, we can upper-bound $\Exp{\Delta(x,y)}$ by:
$2\frac{c}{n}(\frac{\rho}{\sqrt c})^2 + (\frac{n-c}{n})^2(\frac{ \ln n}{n})^2 < 2$.

\textbf{x and y are in the same group:} The third node $z$,
  can be in the same group as $x$ and $y$, can share a group with only
  one of them or be an outsider to both nodes. The upper bound for the
  probability of any of these events happening is $\frac{2c}{n}$,
  $\frac{c}{n}$ and $\frac{n-c}{n}$. So the probability of triangle
  $(x,y,z)$ existing is:
  \begin{center}
$\frac{c}{n}(\frac{\rho}{\sqrt c})^2  + 2\frac{c}{n}\frac{n-c}{n}(\frac{\rho}{\sqrt c})(\frac{ \ln n }{n}) + \frac{(n-c)}{n}\frac{n-2c}{n}(\frac{ \ln n }{n})^2$
\end{center}
By multiplying $(n-2)$ by the value above, we will find an upper-bound for $\Exp{\Delta(x,y)}$ which is smaller than $1$.\\
\textbf{x and y are in different groups:} Each third node
  $z$, can be in the same Type 1 group with $x$ and the same Type 2
  group with $y$ or vice-versa, it can also share a group with either
  one of $x$ or $y$ but not both, or it can be an outsider to both
  nodes. An upper-bound for the probability of having triangle
  $(x,y,z)$ is:
\begin{center}
$2(\frac{c}{n})^2(\frac{\rho}{\sqrt{c}})^2+2\frac{2c}{n} (\frac{n-c}{n})(\frac{\rho}{\sqrt c})(\frac{ \ln n }{n}) + (\frac{(n-2c)}{n})^2(\frac{ \ln n }{n})^2$
\end{center}
Now if we multiply this by $n-2$ it will be the expected number of triangles which is less than $1$.
\hfill $\qed$
\end{proof}

In summary, for the common cases, $E[\Delta(x,y)] < 1$, and even for
the rare case in which $x$ and $y$ have two groups in common, 
$E[\Delta(x,y)] < 2$. Thus, one might not
expect triangles to be a useful feature in discriminating between the
two kinds of ties. 

Now, we present an analysis for the use of squares as features.
\begin{lemma}
  If $x$ and $y$ are in a different group, $\Exp{\Box(x,y)} < 1$ and if they share a group, $\Exp{\Box(x,y)} \geq \sqrt{c}\rho^3$.
\end{lemma}
\begin{proof}
  We will analyze this via a case analysis, as before:
  \\
\textbf{x and y are in the same group:}
  The number
  of squares including $(x,y)$ in this case is lower-bounded by the
  expected number of squares when all four nodes are inside the same
  group. The probability of the other two nodes being in the same
  group is $(\frac{c}{n})^2$ so the probability of a square with the
  specific ordering of $x$, $u$, $v$, $y$ will be $ (\frac{c}{n})^2(\frac{\rho}{\sqrt c})^3 = \frac{\sqrt{c}\rho^3}{n^2}$.\\
   There are \small $(n-2)(n-3)$ \normalsize candidate pairs leading to the claimed lower-bound.
    If these two nodes
    share two groups then with a similar method and using Lemma
    \ref{IntersectionSize} we can get a lower-bound of
    $2\sqrt{c} \rho^3$.
\\
\textbf{x and y are in different groups:}
    When $x$ and $y$ are in different groups, there are two ways a square
can form that includes $(x,y)$.  The first way is for the square to include
an edge whose ends do not share a group.  The second is for each of the four
edges to have ends that share a group.  Note that this second situation
is not possible in the single planted model, but becomes possible in
the double planted model whenever there are two other nodes $u$ and $v$ where
$(x,u)$,$(u,v)$ and $(v, y)$ each have a group in common.
We will call such a pair $(u,v)$ a {\em potential bad pair} for $(x,y)$.

The case in which the square involves an edge whose ends do not share
a group is the easier case; the expected number of such squares is very small
due to the low probability on such cross-group edges, and we omit the
details here.
We now consider the case of a square whose edges lie within groups;
we bound the expected number of such squares by bounding the
number of potential bad pairs for two nodes.



    We know that each of $x$ and $y$ have at most $2c(1+\delta)$
    shared group members.
    A neighbor of $x$ and a neighbor of $y$ should be in the same group
    to create a bad pair. We name the groups that exclude $x$ and $y$,
    $S_{j,1},\ldots,S_{j,\frac{n}{c}-2}$ where $j \in \{1,2\}$ denotes
    which group type it is, and the groups including $x$ and
    $y$ as $S_{1,x},S_{2,x},S_{1,y},S_{2,y}$. We know that the number
    of potential bad pairs is:
\begin{center}
    $ \sum_{j=1}^2 [|S_{j,x} \cap S_{3-j,y}| + \sum_{i=1}^{\frac{n}{c}-2} |S_{j,x} \cap S_{(3-j),i}||S_{j,y} \cap S_{(3-j),i}|]$
\end{center}
    So by Lemma \ref{IntersectionSize} we know, the terms
    $|S_{1,x} \cap S_{2,y}|$ and $|S_{2,x} \cap S_{1,y}|$ are upper
    bounded by a small constant. We also know
    $|S_{j,y} \cap S_{3-j,i}|$ and $|S_{j,x} \cap S_{3-j,i}|$ for
    $1 \leq j \leq 2$, are non-zero for at most $(1+\delta)c$ terms
    and when they are non-zero, By Lemma \ref{IntersectionSize} we
    know that they are bounded by a small constant. Therefore, the
    inner sum will be at most $\mathcal{O}(c)$. Therefore, the
    expected number of potential bad pairs is $\mathcal{O}(c)$ which
    means it is less than $\gamma c$ (where $\gamma << 2\alpha_2^2$)
    with high probability.\footnote{By applying the Chernoff bounds to
      two groups of Type-1 having intersection in a Type-2 group we
      can easily get a bound of $\frac{\alpha_2^2}{2}$} Now that the
    grouping is done we add the edge probabilities. Since each edge is
    inside a group it will exist with probability $\frac{\rho}{\sqrt{c}}$
    which means the expected number of bad squares is $\gamma c(\frac{\rho}{\sqrt{c}})^3 < \frac{2\alpha_2^2 \rho^3}{\sqrt{c}}$.
   This number for a large enough $n$ is smaller than $1$.
   \hfill $\qed$
\end{proof}

The separation between $\Exp{\Box(x,y)}$ and
$\Exp{\Delta(x,y)}$ is as in the single planted model,
and on average, we would expect the squares
feature to be more discriminative compared to the triangles
feature. Hence, even with these simple planted models, we observe a
phenomenon where the sparsity of the triangles feature might limit its
usefulness as a feature compared to squares.

\subsection{Simulation Results}

\begin{figure}[t!]
\centering
\begin{minipage}[l]{\columnwidth}
	\centering
	\includegraphics[width=0.48\linewidth]{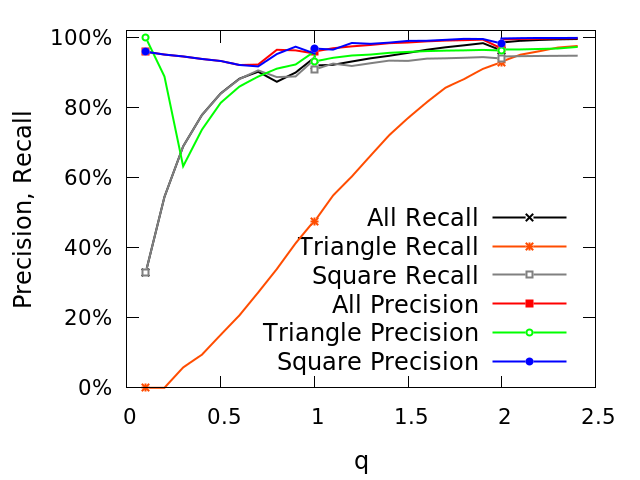}
	\includegraphics[width=0.48\linewidth]{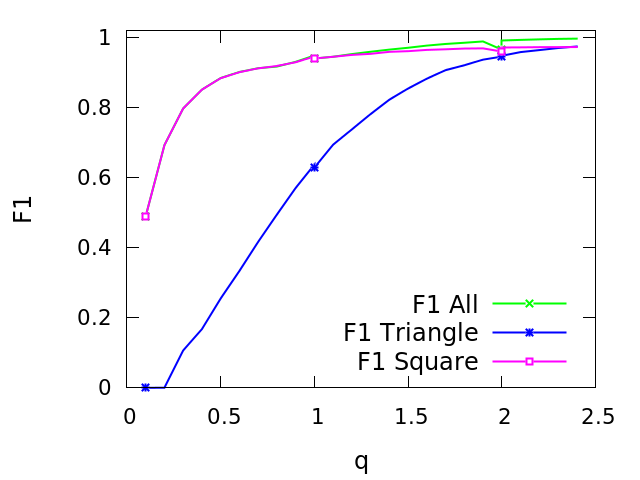}
    \captionsetup{labelformat=empty}
    \caption{(a) Single Planted Model}
\end{minipage}
\begin{minipage}[l]{\columnwidth}
	\centering
	\includegraphics[width=0.48\linewidth]{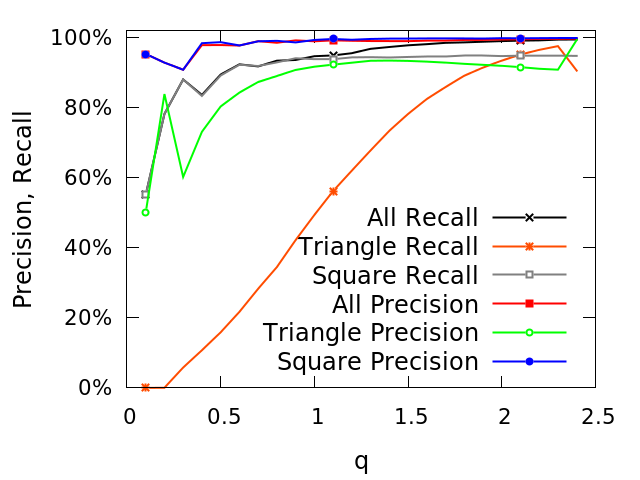}
	\includegraphics[width=0.48\linewidth]{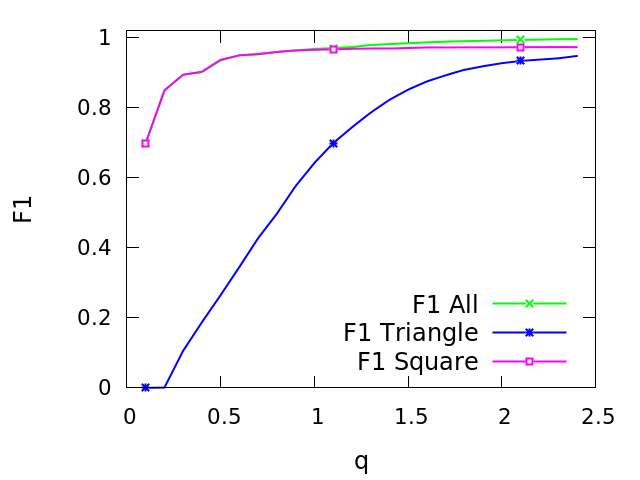}
    \captionsetup{labelformat=empty}
    \caption{(b) Double Planted Model}
\end{minipage}

\caption{Precision, recall and F1 score when $p=0.85$ and variable $q$. (a) Single planted model (b) Double planted model}
\label{Simulation}
\end{figure}

The theoretical results presented above show a separation in the
behavior of the expected number of triangles and squares containing
different kinds of edges.  We now perform computational simulations
to see how these separations in expectation translate into
probabilistic outcomes with concrete values for the parameters.

For all these simulations, we use the prediction framework from
the previous section.
In particular, we use the number of triangles, number of squares, and
their logarithms, both separately and combined together, as
features. Then, as before, we train an LR model that is then used to
make the predictions on the test set.

Given our theoretical results, we expect a parameter range where the
dominant feature changes from triangles to squares. Hence, we run an
experiment for each parameter to understand how this switch
happens. In each of our experiments, we vary the chosen parameter over
a wide range while keeping all the other parameters of the model
fixed. We observe that for most of our parameters, the relative
usefulness of the triangles and squares remains unchanged throughout
the range. The only parameters that have an effect on the features'
relative usefulness are the ones that affect the sparsity of the
resulting graph.

The primary parameter of interest in the planted models is $q$, since
the graph gets sparser as $q$ gets smaller. Hence, in our presentation
below we fix the rest of the parameters and only study the effect of $q$
in the two planted models. We use the specified generative process for
each model to build hundreds of train and test graphs. As suggested by
Figure~\ref{Simulation} the squares features tend to
outperform the triangles when the graph is sparse. In these
simulations we choose $n=4000$, $r=\frac{\ln(n)}{n}$, the expected size for the groups $c$
to be $30$, $p=0.85$ and values in $[0.1, 2.5]$ with a step of $0.1$
for $q$.

\section{Conclusion}

In this work, we use information about a dense graph $G_L$
composed primarily of weak ties to fill in the strong ties in
a sparse graph $G_S$, using the frequency of small subgraphs as features.
We achieve high precision on this prediction task;
however, we need to go beyond structures based on triangles ---
as in standard approaches for strong-tie prediction ---
and look at subgraphs on larger numbers of nodes.
The methodology seems general enough to apply to many
other settings as well.

The relative usefulness of the network structures also presents a
clear trade-off between noisiness and sparsity: from our results it is
evident that small structures that are in direct proximity to the
target node have a higher signal strength than structures that
involve nodes that are farther off. At the same time, these
structures tend to be sparse, and hence in some settings more abundant
structures such as squares are more useful for filling in
a sparsely populated graph.
Our theoretical models and results demonstrate 
contexts where this can be shown to be the case analytically.

The setting discussed here also opens up
several new questions. In particular, is it possible to quantify the
properties of the relationship between $G_L$ and $G_S$ that are
necessary for training on $G_L$ to be useful for filling in $G_S$? We
also note that we experimented using multiple graphs to predict
edges in $G_S$ and we only observed a small ($\sim 1\%$) increase
in precision. 
This leads to the possibility that an understanding of the
overlap properties among graphs might indicate when 
the use of multiple graphs could be effective for these
types of prediction tasks.

\bibliographystyle{abbrv}
\bibliography{refs}  

\begin{thebibliography}{10}

\bibitem{backstrom-cscw14}
L.~Backstrom and J.~M. Kleinberg.
\newblock Romantic partnerships and the dispersion of social ties: A network
  analysis of relationship status on facebook.
\newblock In {\em ACM CSCW}, 2014.

\bibitem{bohb-cold-start}
J.~Bobadilla, F.~Ortega, A.~Hernando, and J.~Bernal.
\newblock A collaborative filtering approach to mitigate the new user cold
  start problem.
\newblock {\em Knowledge-Based Systems}, 26:225--238, 2012.

\bibitem{borgs-graph-hom}
C.~Borgs, J.~T. Chayes, L.~Lovasz, V.~Sos, B.~Szegedy, and K.~Vesztergombi.
\newblock Counting graph homomorphisms.
\newblock In M.~Klazar, J.~Kratochvil, M.~Loebl, J.~Matousek, R.~Thomas, and
  P.~Valtr, editors, {\em Topics in Discrete Mathematics}. Springer, 2006.

\bibitem{condon2001algorithms}
A.~Condon and R.~M. Karp.
\newblock Algorithms for graph partitioning on the planted partition model.
\newblock {\em Random Structures and Algorithms}, 2001.

\bibitem{dean2008mapreduce}
J.~Dean and S.~Ghemawat.
\newblock Mapreduce: simplified data processing on large clusters.
\newblock {\em Communications of the ACM}, 51(1):107--113, 2008.

\bibitem{eagle-inferring-friendship}
N.~Eagle, A.~S. Pentland, and D.~Lazer.
\newblock Inferring friendship network structure by using mobile phone data.
\newblock {\em Proc. Natl. Acad. Sci. USA}, 106(36), 2009.

\bibitem{faust-triad-2010}
K.~Faust.
\newblock A puzzle concerning triads in social networks: Graph constraints and
  the triad census.
\newblock {\em Social Networks}, 32(3), 2010.

\bibitem{flajolet_2008}
P.~Flajolet, {\'E}.~Fusy, O.~Gandouet, and F.~Meunier.
\newblock Hyperloglog: the analysis of a near-optimal cardinality estimation
  algorithm.
\newblock {\em DMTCS Proceedings}, 2008.

\bibitem{gilbert_2012}
E.~Gilbert.
\newblock Predicting tie strength in a new medium.
\newblock In {\em CSCW}. ACM, 2012.

\bibitem{gilbert_2009}
E.~Gilbert and K.~Karahalios.
\newblock Predicting tie strength with social media.
\newblock In {\em SIGCHI}. ACM, 2009.

\bibitem{granovetter-weak-ties}
M.~Granovetter.
\newblock The strength of weak ties.
\newblock {\em American Journal of Sociology}, 78, 1973.

\bibitem{joachims-transductive}
T.~Joachims.
\newblock Transductive learning via spectral graph partitioning.
\newblock In {\em ICML}, 2003.

\bibitem{jones-tie-strength}
J.~J. Jones, J.~E. Settle, R.~M. Bond, C.~J. Fariss, C.~Marlow, and J.~H.
  Fowler.
\newblock Inferring tie strength from online directed behavior.
\newblock {\em PLoS ONE}, Jan. 2013.

\bibitem{kim2011modeling}
M.~Kim and J.~Leskovec.
\newblock Modeling social networks with node attributes using the
  multiplicative attribute graph model.
\newblock In {\em UAI}, 2011.

\bibitem{kim2011network}
M.~Kim and J.~Leskovec.
\newblock The network completion problem: Inferring missing nodes and edges in
  networks.
\newblock In {\em SIAM}, 2011.

\bibitem{kuramochi-freq-subgraph}
M.~Kuramochi and G.~Karypis.
\newblock Frequent subgraph discovery.
\newblock In {\em ICDM}, 2001.

\bibitem{liu-heterogeneous}
L.~Liu, J.~Tang, J.~Han, M.~Jiang, and S.~Yang.
\newblock Mining topic-level influence in heterogeneous networks.
\newblock In {\em CIKM}, 2010.

\bibitem{marsden-tie-strength}
P.~V. Marsden and K.~E. Campbell.
\newblock Measuring tie stength.
\newblock {\em Social Forces}, 63(2), Dec. 1984.

\bibitem{mcsherry-planted}
F.~McSherry.
\newblock Spectral partitioning of random graphs.
\newblock In {\em FOCS}, 2001.

\bibitem{milo-motifs}
R.~Milo, S.~Shen-Orr, S.~Itzkovitz, N.~Kashtan, D.~Chklovskii, and U.~Alon.
\newblock Network motifs: Simple building blocks of complex networks.
\newblock {\em Science}, 298(5594), Oct. 2002.

\bibitem{krzakala2013spectral}
F.~K.~C. Moore, E.~Mossel, J.~Neeman, A.~Sly, L.~Zdeborov{\'a}, and P.~Zhang.
\newblock Spectral redemption in clustering sparse networks.
\newblock {\em Proc. Natl. Acad. Sci. USA}, 110(52), 2013.

\bibitem{mossel2014belief}
E.~Mossel, J.~Neeman, and A.~Sly.
\newblock Belief propagation, robust reconstruction, and optimal recovery of
  block models.
\newblock {\em COLT}, 2014.

\bibitem{spup-cold-start}
A.~I. Schein, A.~Popescul, L.~H. Ungar, and D.~M. Pennock.
\newblock Methods and metrics for cold-start recommendations.
\newblock In {\em Proceedings of the 25th annual international ACM SIGIR
  conference on Research and development in information retrieval}, pages
  253--260. ACM, 2002.

\bibitem{ssbxc-cold-start}
S.~Sedhain, S.~Sanner, D.~Braziunas, L.~Xie, and J.~Christensen.
\newblock Social collaborative filtering for cold-start recommendations.
\newblock In {\em Proceedings of the 8th ACM Conference on Recommender
  systems}, pages 345--348. ACM, 2014.

\bibitem{seshadhri-pinar-kolda-triadic}
C.~Seshadhri, A.~Pinar, and T.~G. Kolda.
\newblock Triadic measures on graphs: The power of wedge sampling.
\newblock In {\em SIAM International Conference on Data Mining (SDM)}, pages
  10--18. SIAM, 2013.

\bibitem{sintos-tsaparas-strong-ties}
S.~Sintos and P.~Tsaparas.
\newblock Using strong triadic closure to characterize ties in social networks.
\newblock In {\em SIGKDD}, 2014.

\bibitem{sun-heterogeneous}
Y.~Sun, R.~Barber, M.~Gupta, C.~C. Aggarwal, and J.~Han.
\newblock Co-author relationship prediction in heterogeneous bibliographic
  networks.
\newblock In {\em International Conference on Advances in Social Networks
  Analysis and Mining, {ASONAM} 2011, Kaohsiung, Taiwan, 25-27 July 2011},
  pages 121--128, 2011.

\bibitem{tang-heterogeneous}
J.~Tang, T.~Lou, and J.~M. Kleinberg.
\newblock Inferring social ties across heterogenous networks.
\newblock In {\em WSDM}, 2012.

\bibitem{ugander-www13}
J.~Ugander, L.~Backstrom, and J.~Kleinberg.
\newblock Subgraph frequencies: Mapping the empirical and extremal geography of
  large graph collections.
\newblock In {\em WWW}, 2013.

\bibitem{wang-advisor-advisee}
C.~Wang, J.~Han, Y.~Jia, J.~Tang, D.~Zhang, Y.~Yu, and J.~Guo.
\newblock Mining advisor-advisee relationships from research publication
  networks.
\newblock In {\em Proceedings of the 16th {ACM} {SIGKDD} International
  Conference on Knowledge Discovery and Data Mining, Washington, DC, USA, July
  25-28, 2010}, pages 203--212, 2010.

\bibitem{yan-han-freq-subgraph}
X.~Yan and J.~Han.
\newblock gspan: Graph-based substructure pattern mining.
\newblock In {\em ICDM}, 2002.

\bibitem{zhu-label-prop}
X.~Zhu and Z.~Ghahramani.
\newblock Learning from labeled and unlabeled data with label propagation.
\newblock Technical report, Carnegie Mellon University, 2002.

\end{thebibliography}
\vfill\eject
\section{Appendix}
\subsection{Computing Squares on Large Graphs} \label{sec:hyperloglog}
Computing the exact number of squares on large graphs can be a
challenging problem. This connects to a large literature on counting
triangles and other network motifs, but the size of the graph in our
problem poses challenges for even the most efficient known
techniques. Here, we briefly point out a scalable heuristic that can
compute the feature using a randomized algorithm by utilizing
HyperLogLog sketches~\cite{flajolet_2008}. This approach does not
provide good worst case guarantees as it is, but in practice we've
observed reasonable performance from the heuristic. We also note that
the counts produced by these sketches are fed to a classifier, which
provides an additional layer of robustness for error.

Consider a node $a$ with direct neighbors $B$ in $G_L$, and second
degree neighbors $C$ (excluding the nodes distance two away that are
in $B$). We need to compute the squares feature for each $b\in B$. A
naive method of computing this would involve forming a list of $b$'s
neighbors that are also in $C$, and doing pairwise intersections
between the lists of all $b$'s. We note that one can flip this
computation around, with each $c\in C$ having a list of neighbors that
is intersected with the set $B$, as then each pair in the resulting
list has a unique square path.\\
This idea can be implemented efficiently with HyperLogLog sketches
that implement intersections~\cite{flajolet_2008}. In particular, each
$c$ computes a sketch of its neighbors, and each $b$ is annotated
with a sketch of $B$'s. Then, for an existing $(b,c)$ edge, the
intersection of these sketches is exactly one more than the number of
squares that $b$ participates in. This computation can be implemented
efficiently in MapReduce~\cite{dean2008mapreduce}.

We re-emphasize that in our experiments in this work, we do not use this
approximation (since the sample size was tractable with exact
computation). We note this algorithm here to demonstrate that the
squares feature can be computed on extremely large graphs in an
efficient manner.
\end{document}